\documentclass[12pt]{article}
\input epsf.tex
\usepackage{amsthm, amssymb, amsmath}
\newtheorem{Theorem}{Theorem}
\newtheorem{lemma}{Lemma}

\newtheorem{remark}{Remark}

\begin{document}

\title{\bf{On the stochastic engine of transmittable diseases in exponentially growing populations}}
\author{Torsten Lindstr\"{o}m \\
Department of Mathematics \\
Linn{\ae}us University \\
SE-35195 V\"{a}xj\"{o}, SWEDEN}
\date{}
\maketitle

\maketitle

\begin{abstract}
The purpose of this paper is to analyze the mechanism for the interplay of deterministic and stochastic models for contagious diseases. Deterministic models for contagious diseases are prone to predict global stability. Small natural birth and death rates in comparison to disease parameters like the contact rate and the removal rate ensures that the globally stable endemic equilibrium corresponds to a tiny average proportion of infected individuals. Asymptotic equilibrium levels corresponding to low numbers of individuals invalidate the deterministic results.

Diffusion effects force frequency functions of the stochastic model to possess similar stability properties as the deterministic model. Particular simulations of the stochastic model predict, however, oscillatory patterns. Small and isolated populations show longer periods, more violent oscillations, and larger probabilities of extinction.

We prove that evolution maximizes the infectiousness of the disease as measured by the ability to increase the proportion of infected individuals. This holds provided the stochastic oscillations are moderate enough to keep the proportion of susceptible individuals near a deterministic equilibrium.

We close our paper with a discussion of the herd-immunity concept and stress its close relation to vaccination-programs.
\end{abstract}



\section{Introduction}
\label{intro}

Already Kermack and McKendrick (1927)\nocite{Kermack.ProcCRSocLonA:115} discovered many properties of the epidemic bell-shaped curve. Their first result was that there exists a threshold density of susceptible individuals such that if the density of susceptible individuals exceeds that threshold, then there will be an epidemic. They used the term excess to describe the difference between the density of susceptible individuals and the threshold density. Their second result states that if the excess is small, then the epidemic curve is symmetric and the eventual number of ill corresponds to the double of the excess (Figure \ref{ev_number_ill}(b)). Such results do not hold for large excesses but their third result states that a very infectious disease exhausts almost the whole population of susceptible individuals. The general result is that the eventual number of ill is an increasing function of the infectiousness of the disease, cf. Figure \ref{ev_number_ill}(a).

\begin{figure}
\epsfxsize=128mm
\begin{picture}(264,155)(0,0)
\put(0,0){\epsfbox{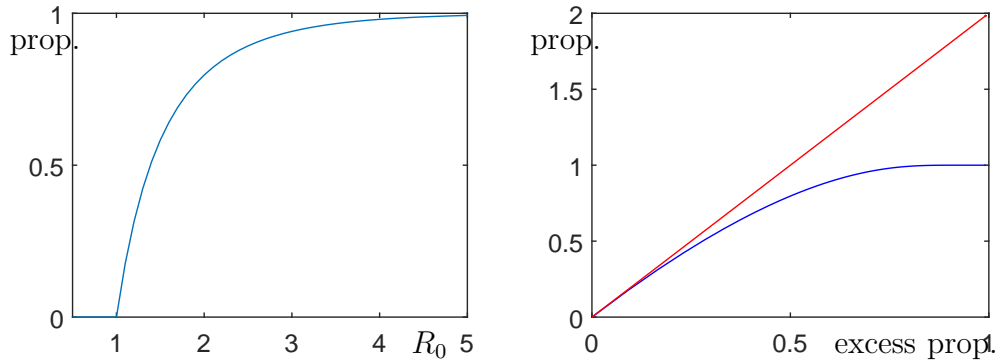}}


\put(145,0){$R_0$}
\put(-7,115){prop.}

\put(305,0){excess prop.}
\put(190,115){prop.}
\end{picture}
\caption{(a) The eventual proportion of ill (blue curve) is an increasing function of $R_0$. (b) The eventual proportion of ill (blue curve) increases with the excess and is approximately the double of the excess (red line) for a small epidemics. The model used for computation is (\protect\ref{Kermack}) with $R_0=\rho_0$.}
\label{ev_number_ill}
\end{figure}

The most well-known measure of the infectiousness of a disease is $R_0$. It is a dimensionless quantity giving the number of secondary infections resulting from each infection in an entirely susceptible population and measures the ability of the disease to increase the number of infected individuals. In many cases, it is the product of the contact rate and the time spent by an infected individual in the population. Evolution of diseases tends to increase a closely related parameter that we baptize as $\rho_0$ (see Section \ref{global-stability}) that measures the ability of the disease to increase the proportion of infected individuals. In order to increase the time that an infected individual spends in the population, we expect diseases to become less virulent with time (see e. g Smith (1887)\nocite{Smith.AMNat:21}). However, evolution maximizes a product and neither the contact rate nor the time in the population. If evolutionary changes that allows an increased contact rate are available, then the pathogen does not need to become less virulent.

Many pre-vaccination pandemic disease control strategies (lock-downs, social distancing, testing, contact tracing, and quarantine restrictions) aim at making the disease less infectious and hence to decrease $R_0$. Kermack and McKendrick (1927)\nocite{Kermack.ProcCRSocLonA:115} predicts that the eventual number of ill decreases as such measures are applied. Testing, contact tracing and quarantine restrictions removes infected individuals from the population and thus, decrease the time spent by an infected individual in the population. Such measures are specific in the sense that they attack a selected specific pathogen. The drawback with these measures is that their efficiency drops rapidly with the size of the epidemic and that a surplus of technical equipment and medical competence does not exist in the beginning of a new pandemic.

Lock-downs and social distancing are general measures for large epidemics. They do not possess the capacity and competence bottlenecks mentioned above but their drawback are the huge costs on society level limiting their application time. In the huge cities there are also bottleneck problems as stairs, elevators and almost no real possibility to move efficiently without public transportation etc. that decrease the
possibilities of implementing any lock-downs efficiently over longer periods. Their objective must be to increase the contact-tracing and testing capacity considerably in order to preserve $R_0$ afterwards. All disease control strategies do not have an impact on $R_0$ alone. Regional quarantine and travel restrictions are such examples. Apart from having an impact on the contact rate, they split the population into smaller units magnifying the probability of local extinction of the disease.

Germann, Kadau, Longini, and Macken (2006)\nocite{Germann.PNAS:103} criticized the classical results and concluded that some of the disease control strategies does just delay the pandemic without affecting the eventual number of ill. We suspect that the authors did never validate these simulation results against Kermack and McKendrick (1927)\nocite{Kermack.ProcCRSocLonA:115} in limiting cases. A far more severe criticism against the classical models is their tendency to predict global stability in contrast to the usually observed outbreaks. These global stability properties of epidemic models are extremely robust. Delays are sources of instabilities for a large number of systems (see e. g Smith (2011)\nocite{smith_delays}) but they are not in general sources of instabilities for epidemic models (McCluskey (2010)\nocite{McCluskey.NA.RWA:11}).

Another possible source of instability is the seasonality of the disease. There is substantial evidence for that many pathogens are more infectious during certain seasons (may be wet, dry, cold, humid, or hot seasons). Evidence for
period doubling routes to chaos (Aron and Schwartz (1984)\nocite{aron}, Glendinning and Perry (1997)\nocite{Glendinning.JoMB:35}, Keeling, Rohani, and Grenfell (2001)\nocite{Keeling.physicaD:148}, and
Olsen and Schaffer (1990)\nocite{Olsen.Science:249}) and deterministic chaos (Barrentos, \'{A}ngel-Rodr\'{i}\-guez, and Ruiz-Herrera (2017)\nocite{Barrentos.JoMB:75}) connected to seasonal epidemic models exists.

A dominant pattern in all epidemic oscillations is, however, the tendency to display patterns that are more erratic and possess longer periods with a lower community size and populations that are more isolated, cf. Bartlett (1957)\nocite{Bartlett.JRSS_A:120}. There are evidence for such patterns for many diseases and we obtain the most striking evidence when comparing long term European or North-American pre-vaccination data to corresponding data from the more isolated population
at Iceland (Cliff and Haggett (1980)\nocite{Cliff.JHygiene:85}, Galazka (1991)\nocite{Galazka.Epi_Inf:107}, T{{\'{o}}}masson and {{\"{O}}}gmundsd{{\'{o}}}ttir (1975)\nocite{Tomasson:ActaPath:83}, Gudmundsd{{\'{o}}}ttir, Antonsd{{\'{o}}}ttir, Gudnad{{\'{o}}}ttir, Elefsen, Einarsd{{\'{o}}}ttir, {{\'{O}}}lafsson, and Gudnad{{\'{o}}}ttir (1985)\nocite{Gudmundsdottir.BoWHO:73}). Deterministic periodic solutions or chaos should be an outcome of a selection of disease
and population parameters that is not related to community size.

The second serious problem with many epidemiological models is the lack of structural stability, cf. Guckenheimer and Holmes (1983)\nocite{guck}. This problem is visible already in the models by Kermack and McKendrick (1927)\nocite{Kermack.ProcCRSocLonA:115} and up to our knowledge this problem has not so far been discussed in the literature with reference to epidemic models. The characteristic of a structurally unstable model is that its predictions are sensitive to various assumptions made in the model. One example is the bell-shaped epidemic curve mentioned above. It has the tendency to approach zero both before and after its maximum giving the impression that an epidemic does not remain in the population after invading it in contrast to the above-mentioned global stability results.

Attempts to remedy the structural stability problem exists in the literature, and already Kermack and McKendrick (1932)\nocite{Kermack.ProcCRSocLonA:138} made an attempt to include both migration and births to their model. The most common way is to include a migration rate in the model (Bailey (1964)\nocite{Bailey.Elements}) but this results in a model that is not closed with respect to the population. A subsequent problem is that the introduced parameters are not easy to compare to each other.

Some authors work with the constant population case, see e. g. Brauer and Castillo-Ch{\'{a}}vez (2001)\nocite{Brauer.Mathmod} and Arino, McCluskey, and van den Driessche (2003)\nocite{Arino.SIAMJAM:64} but this does not remove the structural stability problem. None of these attempts sheds light on the parameter combinations that are typical for epidemic models, ie the disease parameters are acting on a much faster scale than the population parameters. One attempt is also to work with bounded populations (Bremermann and Thieme (1989)\nocite{Bremermann}). Also here, the problem of comparing the introduced parameters in a natural way to each other persists.

Our approach to remedy the structural stability problem is to consider an epidemic model allowing for exponential growth of the population. Previous attempts to use such an approach can be found in e. g Anderson and May (1982)\nocite{Andersson.Parasitology:85}. We perform a transformation into density variables ensuring a precise and complete qualitative analysis on a compact set. In the density formulation, the role of $R_0$ as a threshold parameter is lost. It is replaced by the closely related parameter $\rho_0$ mentioned above. In the constant population case these two parameters are equal, ie $R_0=\rho_0$.

For most diseases, the natural birth and death parameters tend to operate on a slower scale than typical disease parameters like the contact rate and the removal rate. After an epidemic, this parameter asymmetry forces the number of infected individuals to low equilibrium numbers making differential equations approximations of birth-death processes (Bailey (1964)\nocite{Bailey.Elements}) invalid. This constitutes a mechanism for the observed instabilities and explains why not only the model parameters but also the population size plays an important role as a source of instability for infectious diseases. A number of mechanisms that might result in stochastic oscillations in epidemic models possibly together with other mechanisms like seasonal forcing have been listed and analyzed in Black and McKane (2010)\nocite{Black.JRSI:7}. We note that none of these mechanisms is the one announced here and end up with a model that neither have been formulated nor been studied before.

The structure of our paper is as follows: We formulate the deterministic version of the model in Section \ref{modelsection}. We list basic properties of stochastic models in Section \ref{modelsection}, too. We make a transition to density coordinates for the deterministic model in Section \ref{density-coordinates} but still keep our original abundance based model in mind for later analysis. Density coordinates guarantees the existence of a compact invariant set for the deterministic model. We consider the evolutionary development of the involved pathogens in Section \ref{global-stability} and prove a global stability theorem that completes the qualitative analysis of our model. This global stability theorem shows that evolution tends to maximize the newly introduced parameter $\rho_0$. This results is in agreement with existing results asserting that evolution maximizes $R_0$ for bounded populations (Bremermann and Thieme (1989)\nocite{Bremermann}).

We follow up with a discussion of relations between the involved parameters and their consequences for the application of various disease control strategies in Section \ref{real-param}. We formulate the complete stochastic version by returning to the abundance based model from Section \ref{modelsection} in Section \ref{stoch_form}. We analyze it numerically in a number of ways and compare its predicted densities to the deterministic damped oscillations. In Section \ref{vacc}, we discuss vaccination campaigns and their close relation to the herd-immunity concept. Section \ref{sum} contains a summary of our results.

\section{The model}
\label{modelsection}
We consider the following system of nonlinear ordinary differential equations
\begin{eqnarray}
  \dot{S}&=&\lambda N-\beta_1 S\frac{I_1}{N}-\beta_2 S\frac{I_2}{N}-\omega S-\mu S,\nonumber\\
  \dot{I}_1&=&\beta_1 S\frac{I_1}{N}-\gamma_1 I_1 -\mu I_1,\nonumber\\
  \dot{I}_2&=&\beta_2 S\frac{I_2}{N}-\gamma_2 I_2 -\mu I_2,\label{number_model}\\
  \dot{R}&=&\gamma_1 I_1+\gamma_2 I_2+\omega S-\mu R,\nonumber
\end{eqnarray}
and its corresponding stochastic counterparts. Its variables are $S$, the number of susceptible individuals, $I_1$,
the number individuals infected with strain 1, $I_2$, the number individuals infected with strain 2,
and finally, $R$, the number of removed. The removed individuals may consist of either recovered immune individuals or infected individuals in quarantine. The parameters are the birth rate $\lambda$, the contact rate $\beta_1$ for
strain 1, the contact rate $\beta_2$ for strain 2, the mortality $\mu$, the removal rate $\gamma_1$ for strain 1, and the removal rate $\gamma_2$ for strain 2. The final parameter $\omega$ is a vaccination rate enabling the removal of susceptible individuals from the population. The model has many similarities with some of the models studied in Anderson and May (1982)\nocite{Andersson.Parasitology:85}, but our analysis and further extensions of this model are different.

There is a balance between the explanatory value of a model and the possibilities to include phenomena of interest. Our intention is to keep (\ref{number_model}) as restricted as possible in order to allow for an analysis of models that includes population size as an important model parameter later on. We have included two strains in the model at this stage in order to assess the impact of pathogen evolution on the various parameters of the model. As soon as this is known, we remove the second strain. Several strains could, in principle, be included but the subsequent analysis would not alter the results.

We consider neither the possibility of co-infection with both strains, partial or waning immunity, nor an excess mortality due to any or both of the infections. We experience later, that asymptotic proportions of infectious individuals remain at low levels and that even lower average proportions must hold for possible co-infected proportions. We also assume that evolution of new strains operates on a slow time-scale granting at least some partial immunity for possible new strains in the removed class. The reason for assuming no excess mortality is that this assumption makes the model easier to check in limiting cases. Similarly, we assume that the vaccine is effective and leads to complete immunity against both strains. Arino, McCluskey, and van den Driessche (2003)\nocite{Arino.SIAMJAM:64} considered the possibilities for using different vaccination strategies and included vaccine efficiency as a parameter.
We shall see that the model (\ref{number_model}) grants generic structural stability (Guckenheimer and Holmes (1983)\nocite{guck}) anyway. Our results regarding long run dynamical properties can therefore, be generalized to models that include small perturbations in the above mentioned or other neglected directions. Nevertheless, it turns out that the probably much smaller birth-rate $\lambda$ is of substantial importance for structural stability (see Section \ref{real-param}). One of the reasons for considering a structurally stable model that is kept simple is the ongoing CoViD-19 discussion. There are reasons to keep in mind what such models actually state about disease transmission and discuss the interpretation of these results.

If we use the assumption of no excess mortality, the whole population grows exponentially since
\begin{equation}
  \dot{N}=\dot{S}+\dot{I}_1+\dot{I}_2+\dot{R}=\lambda N-\mu S-\mu I_1-\mu I_2-\mu R=(\lambda-\mu)N,
\label{whole-pop-diff-eq}
\end{equation}
with $N(0)=N_0$. Thus, the whole population obeys $N(t)=N_0\exp((\lambda-\mu)t)$. Our assumptions give rise to a closed generically structurally stable prototype model allowing for several precise analytical statements and numerical
results that can be validated in a substantial number of limiting cases.

We shall consider the deterministic model (\ref{number_model}) side by side with its stochastic counterparts. The corresponding stochastic model for the whole population is well-known, see e. g Bailey (1964)\nocite{Bailey.Elements}. Indeed, consider the corresponding birth-death Markov process
with state space ${\bf{N}}=\{0,1,2,\dots\}$. Let $p_N(t)$ be the probability that the population size is $N$ at time $t$. We assume that
\begin{eqnarray*}
  {\rm{Chance \: of \: one \: birth}}&=& \lambda N(t)\delta t+o(\delta t),\\
  {\rm{Chance \: of \: one \: death}}&=& \mu N(t)\delta t+o(\delta t),\\
  {\rm{Chance \: of \: more \: than \: one \: death/birth}}&=& o(\delta t)
\end{eqnarray*}
during a small time interval $\delta t$. The differential-difference equations
\begin{eqnarray}
  \dot{p}_N &=& \lambda(N-1)p_{N-1}-(\lambda+\mu)Np_N+\mu(N+1)p_{N+1},\: N=1,2,3\dots \nonumber\\
  \dot{p}_0 &=& \mu p_1\label{birth-death-diff-diff}
\end{eqnarray}
on ${\bf{R}}^\infty$ gives the probabilities that the population has a given size $N$ at time $t$. If we know that the initial population size is $N_0$, then $p_{N_0}(0)=1$ and $p_N(0)=0$ for $N\neq N_0$. Bailey (1964)\nocite{Bailey.Elements} proved that (\ref{birth-death-diff-diff}) has the explicit solution
\begin{eqnarray}
  p_N(t) &=& \sum_{j=0}^{\min(N_0,N)}\left(\begin{array}{c}N_0\\j\end{array}\right)\left(\begin{array}{c}N_0+N-j-1\\N_0-1\end{array}\right)\alpha^{N_0-j}\beta^{N-j}(1-\alpha-\beta)^j, \nonumber\\
  p_0(t) &=& \alpha^{N_0},\label{extinction_formula}\\
  {\rm{with}} && \alpha =\frac{\mu(1-e^{-(\lambda-\mu)t})}{\lambda-\mu e^{-(\lambda-\mu)t}}\:\:\:\:{\rm{and}}\:\:\:\:
  \beta =\frac{\lambda(1-e^{-(\lambda-\mu)t})}{\lambda-\mu e^{-(\lambda-\mu)t}}\nonumber
\end{eqnarray}
for $N=1,2,3,\dots$ and $N=0$, respectively. The extinction state containing zero individuals is an absorbing state indicated by the fact that the probability of extinction increases with time. The eventual probability of extinction equals one if $\lambda\leq\mu$ and is $(\mu/\lambda)^{N_0}$ when $\lambda>\mu$.
We conclude that the eventual probability of extinction decreases rapidly with initial population size if $\lambda>\mu$. We shall use the above results for validation of our simulation results in limiting cases later on. Bailey (1964)\nocite{Bailey.Elements} also found that
\begin{eqnarray*}
  E(N(t)|N(0)=N_0)&=&N_0\exp((\lambda-\mu)t)\\
  V(N(t)|N(0)=N_0)&=&\frac{N_0(\lambda+\mu)}{(\lambda-\mu)}\exp((\lambda-\mu)t)(\exp((\lambda-\mu)t)-1)
\end{eqnarray*}
meaning that the expected value of the introduced stochastic process coincides with the deterministic value. Unfortunately, such nice results do not hold in many nonlinear cases . Stochastic counterparts of very simple nonlinear models like the e. g Lotka (1925)\nocite{Lotka.Elements} and Volterra (1926)\nocite{Volterra.Mem} do not possess this property, cf. Bharucha-Reid (1960)\nocite{Bharucha-Reid.Elements}.

The differential-difference equations (\ref{birth-death-diff-diff}) must keep the total probability invariant. This is easy to check for (\ref{birth-death-diff-diff}). We add the following lemma with proof because this turns out to be an important technique for validating more complicated sets of differential-difference equations used in the numerical work later on.
\begin{lemma}
The total probability $\sum_{N=0}^{\infty}p_N(t)$ is an invariant for \em (\ref{birth-death-diff-diff}) \em and the total contribution to changes in the total probability related to each of the parameters are zero.
\label{probability_invariance_lemma}
\end{lemma}
\begin{proof}
It is most instructive to divide the computation into
\begin{eqnarray*}
\mu p_1(t)-\mu\sum_{N=1}^\infty Np_N(t)+\mu\sum_{N=1}^\infty(N+1)p_{N+1}(t)&=&\\
\mu p_1(t)-\mu\sum_{N=1}^\infty Np_N(t)+\mu\sum_{M=2}^\infty Mp_{M}(t)&=&0
\end{eqnarray*}
and
\begin{eqnarray*}
\lambda\sum_{N=1}^\infty(N-1)p_{N-1}(t)-\lambda\sum_{N=1}^\infty Np_N(t)&=&\\
\lambda\sum_{M=0}^\infty Mp_{M}(t)-\lambda\sum_{N=1}^\infty Np_N(t)&=&0.
\end{eqnarray*}
Therefore, the contributions to the changes in the total probability are zero term by term for both the death and the birth process.
\end{proof}

The interchange between the birth-death process (\ref{birth-death-diff-diff}) and the differential equation model (\ref{whole-pop-diff-eq}) describes the transition from the stochastic dynamics that is governing for a small number of individuals to the approximating deterministic dynamics ruling the situation when the population is large and the probability of extinction remains negligible. Any understanding of the dynamics of infectious diseases requires both frameworks. One of the objectives of this paper is to make the fundamental parameter asymmetries that force any epidemic into such a model interaction explicit.

\section{A compact simplex}
\label{density-coordinates}

The exponential growth for (\ref{number_model}) alluded to in Section \ref{modelsection} has a drawback. Limit sets of all solutions are not
compact quantities. Many statements regarding limit sets of the corresponding solutions would remain less precise. A transition to a density formulation will cure this problem since it will project all solutions to a compact simplex.

The introduction of density coordinates will not imply that we leave the abundance formulation but merely that we work with these two formulations side by side. The abundance formulation is inherent in the stochastic modeling approach that we introduced in Section \ref{modelsection} and we therefore, return to the abundance formulation in Section \ref{stoch_form}. Before this, we must understand the mechanism behind the transition from deterministic to stochastic dynamics in epidemic models and this transition is most precisely described in density formulation.

We introduce the density coordinates
\begin{displaymath}
x(t)=\frac{S(t)}{N(t)},\:y(t)=\frac{I_1(t)}{N(t)},\:z(t)=\frac{I_2(t)}{N(t)},\:u(t)=\frac{R(t)}{N(t)}
\end{displaymath}
and get through the chain rule
\begin{eqnarray}
  \dot{x}&=&\lambda -\lambda x-\beta_1 xy-\beta_2 xz-\omega x,\nonumber\\
  \dot{y}&=&\beta_1 xy-\gamma_1 y-\lambda y,\nonumber\\
  \dot{z}&=&\beta_2 xz-\gamma_2 z-\lambda z,\label{density-model}\\
  \dot{u}&=&\gamma_1 y+\gamma_2 z+\omega x-\lambda u.\nonumber
\end{eqnarray}
If we now add all equations, we get
\begin{equation}
  \dot{x}+\dot{y}+\dot{z}+\dot{u}=\lambda-\lambda x-\lambda y-\lambda z-\lambda u=\lambda(1-x-y-z-u).
  \label{simplex-de}
\end{equation}
Equation (\ref{simplex-de}) implies that solutions starting at the plane $x+y+z+u=1$ remain in that plane. Moreover, we have $x=0\Rightarrow \dot{x}=\lambda>0$, $y=0\Rightarrow \dot{y}\geq 0$, $z=0\Rightarrow \dot{z}\geq 0$, $u=0\Rightarrow \dot{u}\geq 0$ meaning that all solutions remain in the compact simplex $x+y+x+u=0$, $x\geq\epsilon>0$, $y\geq 0$, $z\geq 0$, $u\geq 0$. The simplex is an invariant compact set containing all solutions of interest for (\ref{density-model}). The inequality $x\geq\epsilon>0$ ensures that an infectious disease can never exhaust the whole population at the same time, a classical result already stated in Kermack and McKendrick (1927)\nocite{Kermack.ProcCRSocLonA:115}.

Since the first three equations do not contain the last variable, it suffices to consider the first three equations. We get
\begin{eqnarray}
  \dot{x}&=&\lambda-\lambda x-\beta_1 xy-\beta_2 xz-\omega x=\beta_1 x\left(\frac{\lambda (1-x)-\omega x}{\beta_1 x}-y\right)-\beta_2 xz,\nonumber\\
  &=&f_i(x)F_i(x)-yf_1(x)-zf_2(x)\nonumber\\
  \dot{y}&=&\beta_1 xy-\gamma_1 y-\lambda y=y\beta_1\left(x-\frac{\gamma_1+\lambda}{\beta_1}\right)=y\psi_1(x),\label{3D-model}\\
  \dot{z}&=&\beta_2 xz-\gamma_2 z-\lambda z=z\beta_2\left(x-\frac{\gamma_2+\lambda}{\beta_2}\right)=z\psi_2(x),\nonumber
\end{eqnarray}
with $f_i(x)=\beta_ix$, $F_i(x)=(\lambda (1-x)-\omega x)/\beta_i x$, $\psi_i(x)=\beta_ix-\gamma_i-\lambda$, $i=1,2$. The introduced notation will save some work. It follows that we can work on the compact invariant simplex $x+y+z\leq 1$, $x\geq\epsilon>0$, $y\geq 0$, $z\geq 0$ from now on.

Note that the natural mortality disappears from the
model in its density formulation. Since the natural mortality removes infected individuals from the population, $R_0$ cannot be a threshold parameter for (\ref{density-model}) and we return later to the threshold parameters of the density-formulated model. The reason is that $R_0$ expresses changes in the numbers of infected individuals and not changes in proportions of infected individuals.

\section{Global stability and evolutionary optimization}
\label{global-stability}

Model (\ref{3D-model}) contained two strains of the same disease with proportions $y$ and $z$ and we did not consider the possibility of coinfection of both strains.
The strains were assumed to have different contact rates $\beta_i$, $i=1,2$ and different removal rates $\gamma_i$, $i=1,2$.  The following theorem grants competitive exclusion (cf. Hardin(1960)\nocite{Hardin.Science:131})
regardless of stable or oscillatory dynamics.
\begin{Theorem}
    If $\gamma_2-\gamma_1>0$ and $\beta_2-\beta_1<\gamma_2-\gamma_1$, then $y$ outcompetes $z$ in \em (\ref{3D-model}).\em
  \label{global_evolution}
\end{Theorem}
\begin{proof}
  The formula for the derivative of a quotient gives
  \begin{displaymath}
  d\frac{\frac{y(t)}{z(t)}}{dt}=\frac{y\psi_1(x)z-z\psi_2(x)y}{z^2}=\frac{y}{z}(\psi_1(x)-\psi_2(x))=\frac{y}{z}((\beta_1-\beta_2)x-(\gamma_1-\gamma_2)).
  \end{displaymath}
  We compute
  \begin{eqnarray*}
    (\psi_1-\psi_2)(0)&=&-\gamma_1+\gamma_2>0\\
    (\psi_1-\psi_2)(1)&=&\beta_1-\gamma_1-\beta_2+\gamma_2>0.
  \end{eqnarray*}
  The conditions of the theorem grant $(\psi_1-\psi_2)(x)>0$ for all $x\in[0,1]$. The reason is that $\psi_1-\psi_2$ is an affine function. Hence, $y/z$ increases along solutions of (\ref{3D-model}). Bounded proportions gives that $z$ must tend to zero.
\end{proof}
We visualize the competitive exclusion region in the $(\beta_2-\beta_1)-(\gamma_2-\gamma_1)$-plane cf. Figure \ref{val_regio}.
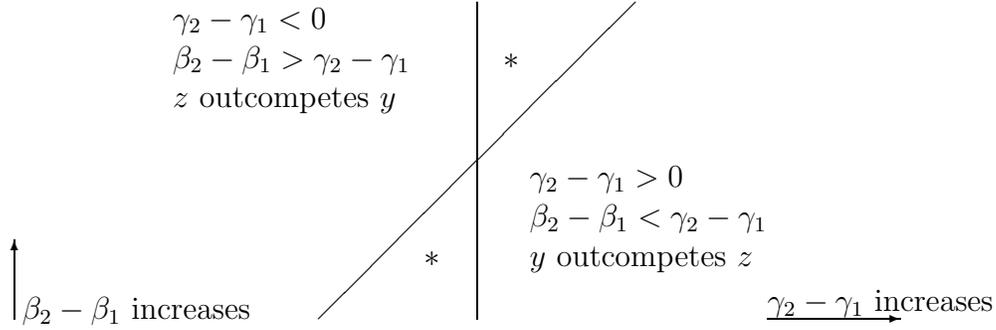
\begin{figure}
  \begin{picture}(360,150)(0,0)
    \put(180,0){\line(0,1){120}}
    \put(120,0){\line(1,1){120}}
    \put(200,50){$\gamma_2-\gamma_1>0$}
    \put(200,35){$\beta_2-\beta_1<\gamma_2-\gamma_1$}
    \put(200,20){$y$ outcompetes $z$}
    \put(160,20){$\ast$}
    \put(65,110){$\gamma_2-\gamma_1<0$}
    \put(65,95){$\beta_2-\beta_1>\gamma_2-\gamma_1$}
    \put(65,80){$z$ outcompetes $y$}
    \put(190,95){$\ast$}
    \put(290,0){\vector(1,0){50}}
    \put(290,3){$\gamma_2-\gamma_1$ increases}
    \put(5,0){\vector(0,1){30}}
    \put(8,0){$\beta_2-\beta_1$ increases}
  \end{picture}
  \caption{Validity region of Theorem \protect\ref{global_evolution}. Competitive exclusion might be violated in the $\ast$-marked regions despite the global stability principle in Theorem \protect\ref{glob-stab-thm}. We discuss the reasons in the subsequent sections.}
  \label{val_regio}
\end{figure}

Model (\ref{3D-model}) has either one, two, or three equilibria. The disease-free equilibrium always exists at $(\lambda/(\lambda+\omega),0,0)$. It is the only equilibrium if
\begin{displaymath}
x_1=\frac{\gamma_1+\lambda}{\beta_1}\geq \frac{\lambda}{\lambda+\omega}, \:\:{\rm{and}} \:\: x_2=\frac{\gamma_2+\lambda}{\beta_2}\geq  \frac{\lambda}{\lambda+\omega}.
\end{displaymath}
In this case, the disease-free equilibrium is globally asymptotically stable since we have $\dot{y}\leq 0$ and $\dot{z}\leq 0$ for valid proportions and it is the only invariant set in the domains $\dot{y}=0$ and $\dot{z}=0$, respectively.
Additional endemic equilibria exist at $(x_1,F_1(x_1),0)$ when $0<x_1<\lambda/(\lambda+\omega)$ and at $(x_2,0,F_2(x_2))$ when $0<x_2<\lambda/(\lambda+\omega)$. We note that if the conditions of Theorem \ref{global_evolution} hold, then $\psi_1-\psi_2$ must be positive
in a neighborhood of $[x_1,x_2]$ and as stated, grant competitive exclusion in equilibrium conditions. It is possible to prove more and indeed, the following global stability theorem holds.
\begin{Theorem}
  Assume that $x_1<\min(x_2,\lambda/(\lambda+\omega))$. Then the endemic equilibrium $(x_1,F_1(x_1),0)$ is globally asymptotically stable in the positive octant and the positive $xy$-plane for \em (\ref{3D-model}). \em
  \label{glob-stab-thm}
\end{Theorem}
\begin{proof}
  Consider the Lyapunov function (cf. Lindstr\"{o}m (1994)\nocite{konflul})
  \begin{equation}
    W(x,y,z)=x_2\left(\int_{x_1}^x\frac{\psi_1(x^\prime)}{f_1(x^\prime)}dx^\prime+\int_{F_1(x_1)}^{y}\frac{y^\prime-F_1(x_1)}{y^\prime}dy^\prime\right)+x_1\int_0^zdz^\prime.
    \label{lyap-fcn}
  \end{equation}
  We compute
  \begin{eqnarray*}
    \dot{W}&=&x_2\psi_1(x)(F_1(x)-y)-x_2\frac{\psi_1(x)}{f_1(x)}f_2(x)z+x_2(y-F_1(x_1))\psi_1(x)+x_1z\psi_2(x)\\
    &=&x_2\psi_1(x)(F_1(x)-F_1(x_1))-x_2\frac{\psi_1(x)}{f_1(x)}f_2(x)z+x_1z\psi_2(x)\\
    &=&x_2\psi_1(x)(F_1(x)-F_1(x_1))+\frac{z}{f_1(x)}\left(x_1\psi_2(x)f_1(x)-x_2f_2(x)\psi_1(x)\right)\\
    &=&x_2\psi_1(x)(F_1(x)-F_1(x_1))+\frac{z\beta_1\beta_2x}{f_1(x)}\left(x_1(x-x_2)-x_2(x-x_1)\right)\\
    &=&x_2\psi_1(x)(F_1(x)-F_1(x_1))+\frac{z\beta_1\beta_2x^2}{f_1(x)}\left(x_1-x_2\right)\leq 0
  \end{eqnarray*}
  The last inequality holds since $F_1$ is decreasing. Global stability follows by LaSalle's (1960)\nocite{LaSalle.IRETCT:7} invariance theorem.
\end{proof}
\begin{remark}\em
It follows that the endemic equilibrium $(x_2,0,F_2(x_2))$ is globally asymptotically stable in the positive octant and the positive $xz$-plane if $x_2<\min(x_1,\lambda/(\lambda+\omega))$ for (\ref{3D-model}) by the same argument.
\em \end{remark}
We could end our analysis at this point since we know all qualitative properties of the deterministic system (\ref{3D-model}).
Evolution selects the strain minimizing $x_i$, $i=1,2$ and competitive exclusion of the model follows.  Yet, we shall see that the global stability theorem above is not sufficient for granting neither competitive exclusion nor exclusion of oscillations (see Section \ref{stoch_form}).
Precisely as rare violations of the competitive exclusion principle might
exist (Lindstr\"{o}m (1999)\nocite{szeged1} and references therein), violations could exist here, in principle.

However, deterministic violations can never occur since we proved  the global stability of the endemic equilibrium. Subsequent waves of the disease should decrease in amplitude with respect to the distance measure defined by the Lyapunov function (\ref{lyap-fcn}). The good news are that the mechanism for stochastic oscillations that we are going to detect later will be prone to create violent oscillations among the infected individuals and not in the susceptible individuals. Violations of the above competitive exclusion principle requires violent oscillations among the susceptibles.

We remark that analogous results have been stated for bounded populations (Bremermann and Thieme (1989)\nocite{Bremermann}). However, the hypothetic possibility for violations of the corresponding competitive exclusion principle through stochastic oscillations due small equilibrium densities of the population of infected individuals has not been made visible.

\section{A huge parameter asymmetry}
\label{real-param}

An important conclusion of the previous section was that it suffices to consider the two-dimensional system
\begin{eqnarray}
  \dot{x}&=&\lambda-\lambda x-\beta xy-\omega x=\beta x\left(\frac{\lambda (1-x)-\omega x}{\beta x}-y\right)=f(x)(F(x)-y),\nonumber\\
  \dot{y}&=&\beta xy-\gamma y-\lambda y=y\beta\left(x-\frac{\gamma+\lambda}{\beta}\right)=y\psi(x)\label{2D-model}
\end{eqnarray}
with $f(x)=\beta x$, $\psi(x)=\beta x -\gamma-\lambda$, and
\begin{equation}
  F(x)=\frac{\lambda(1-x)-\omega x}{\beta x}
  =\frac{\lambda+\omega}{\beta}\left(\frac{\frac{\lambda}{\lambda+\omega}-x}{x}\right)=\frac{\lambda+\omega}{\beta}\left(\frac{\hat{x}-x}{x}\right)
\label{hyperbola}
\end{equation}
in most cases. The results of Section \ref{global-stability} implies that evolution tends to maximize the dimensionless parameter
\begin{displaymath}
\rho_0=\frac{\beta}{\lambda+\gamma}.
\end{displaymath}
A related dimensionless quantity is $R_0$, the basic reproduction number. In (\ref{number_model}) with just one strain (cf. (\ref{number_model_red})), it is defined by
\begin{displaymath}
R_0=\frac{\beta}{\gamma+\mu}=\beta\cdot\frac{1}{\gamma+\mu}.
\end{displaymath}
Taking expectations of the exponential distribution gives that the expected time infected individuals remain infected in the modeled population is given by $1/(\gamma+\mu)$. Either natural death or recovery/quarantine can remove
infected individuals and these two processes occur independently from each other. If disease control measures are applied, then contact tracing and quarantine increase the removal rate.

Putting $R_0>1$ implies that the number of infected individuals increases in a completely susceptible population. It is not equivalent to an increasing proportion of infected individuals.
We illustrate one case with $\rho_0\leq \hat{x}$ ($\hat{x}$ was introduced in (\ref{hyperbola})) in Figure \ref{picture_series1}(a) and one case with $\rho_0>\hat{x}$ in Figure \ref{picture_series1}(b). In the pre-vaccination case these cases correspond to $\rho_0\leq 1$ and $\rho_0> 1$, respectively.

\begin{figure}
\epsfxsize=128mm
\begin{picture}(264,310)(0,0)
\put(0,0){\epsfbox{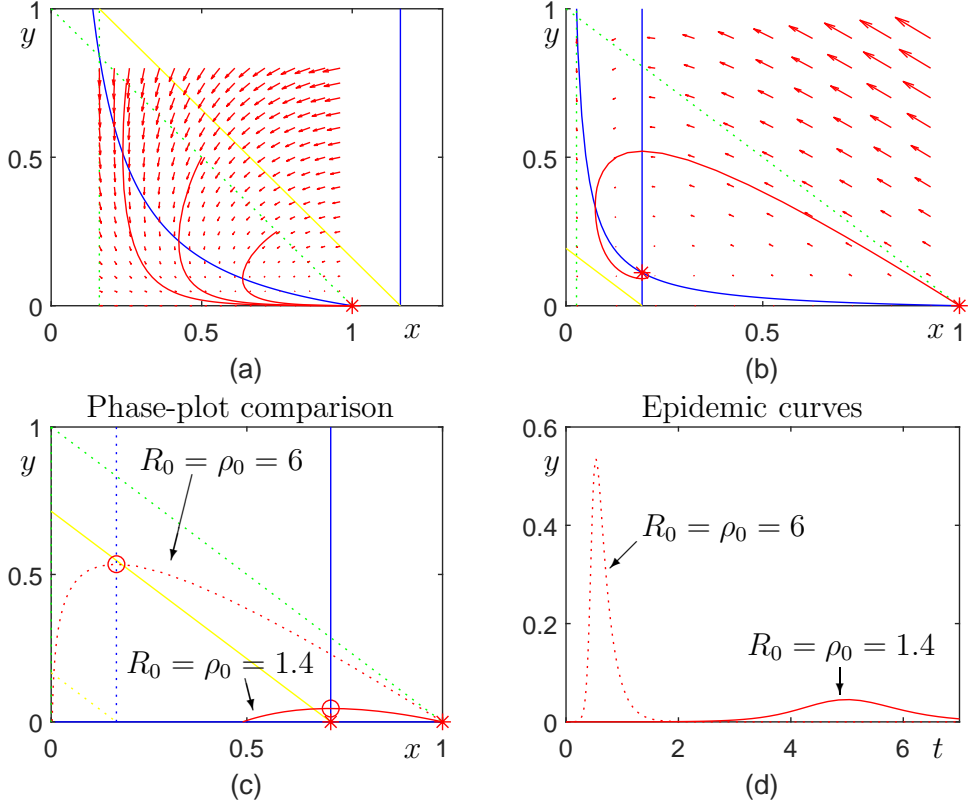}}
\put(150,175){$x$}
\put(5,287){$y$}

\put(348,175){$x$}
\put(203,287){$y$}%

\put(150,16){$x$}
\put(5,125){$y$}
\put(30,145){Phase-plot comparison}
\put(50,125){$R_0=\rho_0=6$}
\put(70,123){\vector(-1,-4){8}}
\put(45,48){$R_0=\rho_0=1.4$}
\put(90,46){\vector(1,-4){3}}

\put(350,16){$t$}
\put(203,125){$y$}
\put(240,145){Epidemic curves}
\put(240,100){$R_0=\rho_0=6$}
\put(238,98){\vector(-1,-1){10}}
\put(280,55){$R_0=\rho_0=1.4$}
\put(315,50){\vector(0,-1){10}}
\end{picture}
\caption{(a) The disease-free equilibrium of (\protect\ref{2D-model}) is globally asymptotically stable for $\lambda=.8$, $\omega=0$, $\beta=5$, and $\gamma=5$. The herd-immunity region (yellow triangle) is automatically approached and the herd-immunity concept is irrelevant. (b) The endemic equilibrium of (\protect\ref{2D-model}) is globally asymptotically stable for $\lambda=.8$, $\omega=0$, $\beta=30$, and $\gamma=5$. The herd-immunity region (yellow triangle) is not a domain of attraction but transients might approach it occasionally. The herd-immunity concept does not make sense. (c) The epidemic in (\protect\ref{Kermack}) exhausts almost the whole population for $\lambda=0$, $\omega=0$, $\beta=30$, and $\gamma=5$ (dotted red curve) whereas approximately half of the population is affected after the epidemic for $\lambda=0$, $\omega=0$, $\beta=7$, and $\gamma=5$ (solid red curve). We approach the corresponding herd-immunity regions (dotted and solid yellow triangles), but the model is not structurally stable. (d) The proportion of infected individuals against time  for $\lambda=0$, $\omega=0$, $\beta=30$, and $\gamma=5$ (dotted red curve) and for $\lambda=0$, $\omega=0$, $\beta=7$, and $\gamma=5$ (solid red curve). }
\label{picture_series1}
\end{figure}
We now prove that small perturbations in the modeling assumptions cannot alter the qualitative dynamical properties of (\ref{2D-model}). In this sense, the relevant parameters are included in the model. We remark that the proof given below can be extended in order to prove that the system (\ref{3D-model}) is generically structurally stable, too.
\begin{Theorem}
Model \em (\ref{2D-model}) \em is structurally stable when $\lambda>0$ and $\rho_0\neq\hat{x}$.
\end{Theorem}
\begin{proof}
The condition $\rho_0\neq\hat{x}$ grants that the fixed points $(\hat{x},0)$ and
\begin{displaymath}
\left(\frac{1}{\rho_0},F\left(\frac{1}{\rho_0}\right)\right)
\end{displaymath}
are hyperbolic when they exist. Since one of these fixed points is globally stable in the positive quadrant, saddle connections can exist only when $(\hat{x},0)$ is a saddle. In this case the positive axis at $y=0$ is the stable manifold of $(\hat{x},0)$. Uniqueness of solutions prevents any solution belonging to the unstable manifold of $(\hat{x},0)$ from intersecting the stable manifold at $y=0$. The global stability properties grants that the system (\ref{2D-model}) is Morse-Smale. Structural stability is a consequence of Palis' and Smale's theorem (1970)\nocite{Palis.AMS_Proc:14}, cf. Guckenheimer and Holmes (1983)\nocite{guck}.
\end{proof}
Is there any reason to continue an analysis of (\ref{2D-model}) at a stage when we know how its parameters evolve with evolution and when we know all its qualitative properties? In fact, our study begins here. Let us focus on the pre-vaccination stage until we reach Section \ref{vacc}. Thus, $\omega=0$. Assume now that the birth rate $\lambda$ remains much smaller than the parameters related to disease transmission. These are the contact rate $\beta$ and the removal rate $\gamma$. We start checking the case $\lambda=0$. In this case we get $\rho_0=R_0=\beta/\gamma$ and our system (\ref{2D-model}) takes the form
\begin{eqnarray}
  \dot{x}&=&-\beta xy,\nonumber\\
  \dot{y}&=&\beta y\left(x-\frac{1}{\rho_0}\right).
  \label{Kermack}
\end{eqnarray}
Many results depend just on the parameter $\rho_0$ now because introduction of dimensionless time removes the second parameter $\beta$, cf e. g Figure \ref{ev_number_ill}.
The problem with (\ref{Kermack}) is that it ceases to be structurally stable (Guckenheimer and Holmes (1983)\nocite{guck}). The reason is the set of neutrally stable fixed points at $y=0$. All results are sensitive to various assumptions made in the model. A maximum for the proportion of infected individuals occurs for $x=\gamma/\beta=1/\rho_0$. If $1/\rho_0$ is just slightly below $1$ then we may call the proportion $1-1/\rho_0$ the excess density as in Kermack and McKendrick (1927)\nocite{Kermack.ProcCRSocLonA:115}. We note that in this case the eventual proportion of ill is approximately the double of the excess (cf. solid red phase curve in Figure \ref{picture_series1}(c) and solid red epidemic curve in Figure \ref{picture_series1}(d)). This confirms classical results that disease control measures decreases the eventual number of ill (cf. Figure \ref{ev_number_ill}). Disease control measures do not just flatten the epidemic curve in order to avoid peak healthcare efforts but save a substantial part of the population from encountering
the disease at all, too. If $1/\rho_0$ is close to zero then the disease affects a large proportion of the population (cf. dotted red phase curve in Figure \ref{picture_series1}(c) and dotted red epidemic curve in Figure \ref{picture_series1}(d)). The reason for that the solid red epidemic curve develops slower in Figure \ref{picture_series1}(d) than the dotted epidemic curve is in Figure \ref{picture_series1}(c): The solid curve remains closer to the equilibria at the horizontal axis of (\ref{Kermack}) and the system is continuous.

The birth rate is certainly not zero but it is in many cases a very small quantity in comparison to the contact rate and the removal rate. We therefore, assume that
\begin{equation}
0<\mu<\lambda<<\gamma<\beta
\label{ess-ineq}
\end{equation}
We keep the natural death rate in this important assumption since it will reappear in the stochastic version of the model. The second inequality in (\ref{ess-ineq}) is needed to grant a positive survival probability of the population (cf. Section \ref{modelsection}) and structural stability of the submodel for the population.

The last inequality in (\ref{ess-ineq}) is satisfied in the most interesting case
$\rho_0>1$ allowing for an invasion of the disease in the absence of a vaccine. The equilibrium proportion of infected individuals is then given by
\begin{eqnarray*}
y_\ast&=&F\left(\frac{\lambda+\gamma}{\beta}\right)=\frac{\lambda+\omega}{\beta}\left(\frac{\frac{\lambda}{\lambda+\omega}-\frac{\lambda+\gamma}{\beta}}{\frac{\lambda+\gamma}{\beta}}\right)
=\frac{\lambda+\omega}{\beta}\left(\frac{\frac{\lambda}{\lambda+\omega}}{\frac{\lambda+\gamma}{\beta}}-1\right)\\
&=&\frac{\lambda}{\lambda+\gamma}-\frac{\lambda}{\beta}-\frac{\omega}{\beta}\leq\frac{\lambda}{\beta}\left(\frac{\beta}{\lambda+\gamma}-1\right)
\end{eqnarray*}
With $\lambda<<\gamma<\beta$ this means that the equilibrium proportion of infected individuals is a very small proportion. If the population is isolated at an island, it might consist of just a few individuals. The more isolated the population is,
the less valid the differential equation model (\ref{number_model}) is and a modeling technique based on stochastic processes like the one mentioned in (\ref{birth-death-diff-diff}) and Bailey (1964)\nocite{Bailey.Elements} is
more appropriate. The assumption of non-zero birth rates change the epidemic curves in Figure \ref{picture_series1}(d) considerably. Relaxing disease control measures has the consequence that the average proportion of infected individuals establishes itself at a considerably higher level after the first wave; see Figure \ref{picture_series2}(a). The reason is that the hyperbola $y=F(x)$, cf. (\ref{hyperbola})
is strictly decreasing (the blue curve in Figures \ref{picture_series1}(a)-(b)) for positive birth rates. The consequence is that the asymtotic proportion of infected individuals increases with relaxed disease control measures. This observation was quite clear in the case of Sweden in comparison to its neighboring countries during the CoViD-19 pandemic. Although the birth rate remains negligible in comparison to the disease parameters, assuming zero birth rates will hide such information.

We have used unusually high birth rates in the figures of this section in order to make key patterns visible. Realistic parameter values would force the hyperbola $y=F(x)$ to approach a degenerate form close to the coordinate axes. Nevertheless, it remains a decreasing and concave up function of $x$ giving support for low average proportions of infected individuals. This is the key observation.

\begin{figure}
\epsfxsize=128mm
\begin{picture}(264,155)(0,0)
\put(0,0){\epsfbox{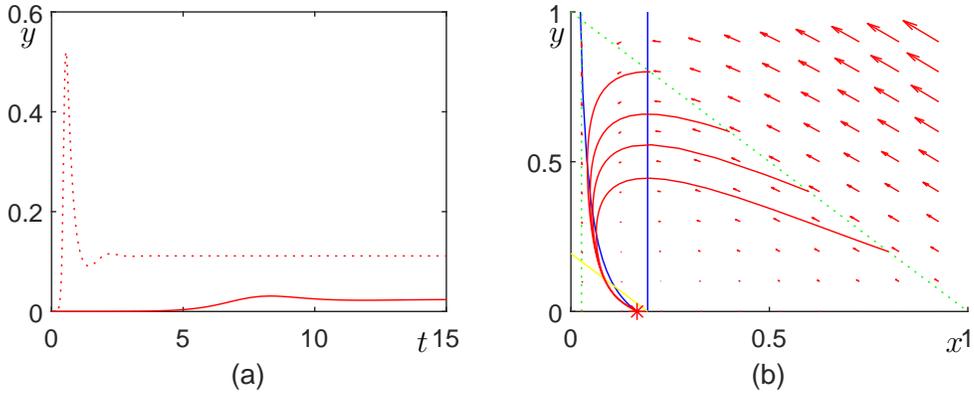}}


\put(155,15){$t$}
\put(5,132){$y$}

\put(355,15){$x$}
\put(205,132){$y$}
\end{picture}
\caption{(a) The proportion of infected individuals against time  for $\lambda=.8$, $\beta=30$, and $\gamma=5$ (dotted red curve) and for $\lambda=.8$, $\beta=7$, and $\gamma=5$ (solid red curve). (b) A sufficient vaccination as is the case for $\lambda=.8$, $\beta=30$, $\gamma=5$, and $\omega=4$ ensures a globally stable equilibrium in the herd immunity region.} \label{picture_series2}
\end{figure}

\section{On the stochastic formulation}
\label{stoch_form}

In the deterministic case we received complete qualitative results for our model after transforming its variables to proportions instead of abundances. The result was that the asymptotic proportion of infected individuals was very low possibly corresponding to just a few individuals. Such situations must be described by the stochastic models introduced in Section \ref{modelsection}. One problem is that this model is inherently abundance formulated and not proportion formulated. We therefore, start the build-up of a stochastic model
from the abundance formulated model we started from in (\ref{number_model}) with one strain. Thus, we use
\begin{eqnarray}
  \dot{S}&=&\lambda N-\beta S\frac{I}{N}-\omega S-\mu S,\nonumber\\
  \dot{I}&=&\beta S\frac{I}{N}-\gamma I -\mu I,\label{number_model_red}\\
  \dot{R}&=&\gamma I+\omega S-\mu R,\nonumber
\end{eqnarray}
as a starting point and let $q_{S,I,R}$ denote the probability that the number of susceptible is $S$, the number of infected is $I$, and the number of removed is $R$. We keep in mind that evolution tends to maximize $\rho_0=\beta/(\lambda+\gamma)$ in most cases. The differential equations for these probabilities then read
\begin{eqnarray}
  \dot{q}_{S,I,R}&=&\mu((S+1)q_{S+1,I,R}+(I+1)q_{S,I+1,R}+(R+1)q_{S,I,R+1})\nonumber\\
  &&-\mu(S+I+R)q_{S,I,R}+\lambda(S+I+R-1)q_{S-1,I,R}\nonumber\\
  &&-\lambda(S+I+R)q_{S,I,R}+\gamma(I+1)q_{S,I+1,R}-\gamma I q_{S,I,R}\nonumber\\
  &&+\beta(S+1)\frac{I-1}{S+I+R}q_{S+1,I-1,R}-\beta S\frac{I}{S+I+R}q_{S,I,R},\nonumber\\
  &&+\omega(S+1)q_{S+1,I,R}-\omega S q_{S,I,R}\nonumber\\
  && \:S=1,2,\dots,\:I=1,2,\dots,\:R=0,1,2,\dots\nonumber\\
  \dot{q}_{0,I,R}&=&\mu(q_{1,I,R}+(I+1)q_{0,I+1,R}+(R+1)q_{0,I,R+1})-\mu(I+R)q_{0,I,R}\nonumber\\
  &&-\lambda(I+R)q_{0,I,R}+\gamma(I+1)q_{0,I+1,R}-\gamma I q_{0,I,R}\label{epid-diff-diff}\\
  &&+\beta\frac{I-1}{I+R}q_{1,I-1,R}+\omega q_{1,I,R},\:S=0,\: I=1,2,\dots,\: R=0,1,2\dots \nonumber\\
  \dot{q}_{S,0,R}&=&\mu((S+1)q_{S+1,0,R}+q_{S,1,R}+(R+1)q_{S,0,R+1})\nonumber\\
  &&-\mu(S+R)q_{S,0,R}+\lambda(S+R-1)q_{S-1,0,R}-\lambda(S+R)q_{S,0,R}\nonumber\\
  &&+\gamma q_{S,1,R}+\omega(S+1)q_{S+1,0,R}-\omega Sq_{S,0,R},\nonumber\\
  &&\:S=1,2,\dots,\: I=0,\: R=0,1,2,\dots\nonumber\\
  \dot{q}_{0,0,R}&=&\mu(q_{1,0,R}+q_{0,1,R}+(R+1)q_{0,0,R+1})-\mu Rq_{0,0,R}\nonumber\\
  &&-\lambda Rq_{0,0,R}+\gamma q_{0,1,R}+\omega q_{1,0,R},\:S=0,\:I=0,\: R=0,1,2,\dots\nonumber
\end{eqnarray}
and we implemented these equations for simulation studies in order to find out the dependence of the size of the population on the results. We have not been able to find a model similar to (\ref{epid-diff-diff}) anywhere in the
literature. We will present the subsequent simulation results in proportion form. This grants that the simulation results of (\ref{epid-diff-diff}) are directly comparable to corresponding simulations of (\ref{2D-model}).
We identify possible limiting cases before continuing with simulation studies. We start with the following theorem.
\begin{Theorem}
The total probability $\sum_{S=0}^{\infty}\sum_{I=0}^{\infty}\sum_{R=0}^{\infty}q_{S,I,R}(t)$ is an invariant for \em (\ref{epid-diff-diff}) \em and the total contribution to changes in the total probability related to each of the parameters is zero.
\label{probability_invariance_thm}
\end{Theorem}
The proof of this theorem is based on the techniques alluded to in Lemma \ref{probability_invariance_lemma}.
We also conclude that extinction of the infection is an absorbing state and that extinction of the population is an absorbing state within that state. The proof follows standard techniques and we give it in Appendix \ref{Appendproof_absorbing}.
\begin{Theorem}
The probabilities $q_{0,0,0}(t)$ and $\sum_{S=0}^\infty\sum_{R=0}^\infty q_{S,0,R}(t)$ increase with time.
\label{absorbing_state_theorem}
\end{Theorem}
The birth-death processes of model (\ref{epid-diff-diff}) coincide with those of the birth-death model (\ref{birth-death-diff-diff}). Therefore, solutions of (\ref{epid-diff-diff}) coincide with the known solutions (\ref{extinction_formula}) of (\ref{birth-death-diff-diff}) in a limiting case. This is an important part of the validation of the numerical version of (\ref{epid-diff-diff}).

We illustrate typical simulation results in Figure \ref{probabilities}. We use $\lambda=.0051$, $\mu=.005$, $\gamma=.2$, $\omega=0$ and $\beta=1.5$ to illustrate a typical outcome and the initial conditions are $S(0)=209$, $I(0)=1$, and $R(0)=0$. This selection gives rise to just negligible probabilities corresponding to population sizes above $300$ as long as we limit the simulation to the time interval $0\leq t\leq 200$. Our numerical model for (\ref{birth-death-diff-diff}) contains $4,545,100$ equations as this size limit is applied. This allowed for
completing the simulation on a laptop during a few days without encountering memory problems.

The thin green curve on the top of the diagram illustrates the $\log_{10}$-value of the survival probability of the population indicating that probability of survival for the population is almost 1. The thick green curve illustrates the $\log_{10}$-value of the probability of survival of the infection. This probability drops rapidly after $t=150$ which makes it hard to do particular simulations of (\ref{epid-diff-diff}) that last longer than this
(cf. Bartlett (1957)\nocite{Bartlett.JRSS_A:120} and Black (1966)\nocite{Black.JTB:11}). We denote the expectations of the proportions of the susceptible individuals, infected individuals, and removed individuals by solid curves in blue, red, and yellow, respectively. These values have no tendency to show oscillatory behavior and we expect that this behavior to be even more robust for larger populations that have more valid diffusion approximations, see e. g Allen (2008)\nocite{Brauer_Driesche_cpt3}. Possible instabilities have therefore, not their origin in the time-dependent probabilities the model (\ref{epid-diff-diff}) themselves.

The instabilities observed for the red curve in Figure \ref{probabilities} as the expected proportion of infected individuals reaches values about $10^{-3}$ is a numerical phenomenon. It originates from the fact that a large number of the probabilities approach zero. The numerical model might suggest tiny negative probabilities that we replaced by zeros.

We used dashed curves in the corresponding colors to illustrate the corresponding predictions of the deterministic model (\ref{number_model_red}) in the same diagram. According to Theorem \ref{glob-stab-thm}, they predict damped oscillations. We indicated one particular simulation of (\ref{epid-diff-diff}) up to extinction by dotted curves in the corresponding colors. In general, extinction of the disease occurs before the drop in the survival probability of the disease takes place after $t=150$. We conclude that the fluctuations especially in the fraction of infected individuals are quite violent before extinction of the disease. This indicates that the instability is a feature of the particular realizations of orbits and not the frequency distributions that act as solutions to (\ref{epid-diff-diff}).
\begin{figure}
\epsfxsize=128mm
\begin{picture}(264,270)(0,0)
\put(0,0){\epsfbox{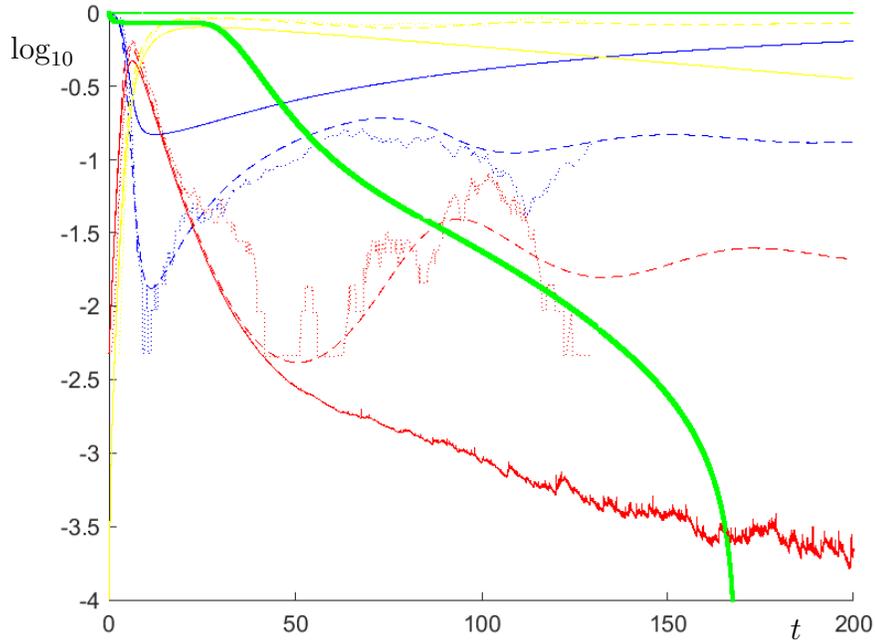}}
\put(10,235){$\log_{10}$}
\put(305,15){$t$}
\end{picture}
\caption{$\log_{10}$ values of either proportions or probabilities. Model (\protect\ref{epid-diff-diff}) is compared to its deterministic counterpart (\protect\ref{number_model_red}) and the parameters are $\lambda=.0051$, $\mu=.005$, $\gamma=.2$,  $\omega=0$ and $\beta=1.5$. The initial conditions are $S(0)=209$, $I(0)=1$, and $R(0)=0$. The thin and thick green curves are the survival probabilities of the population and the infection, respectively. The solid blue, red, and yellow curves are the expected proportions of susceptible, infected and removed individuals according to (\protect\ref{epid-diff-diff}). The dashed curves are the corresponding simulation of the deterministic model (\ref{number_model_red}). We indicate a particular realization of an outcome from model (\protect\ref{epid-diff-diff}) by dotted curves.}
\label{probabilities}
\end{figure}

In Figure \ref{probabilities_mc}, we dropped the simulations of the probabilities of (\ref{epid-diff-diff}) for a case where the model has grown out of bounds. Here we have $\lambda=.0051$, $\mu=.005$, $\gamma=.2$, $\omega=0$, and $\beta=1.5$ but now the initial conditions are $S(0)=1499$, $I(0)=1$, and $R(0)=0$. However, we still simulated particular simulations up to extinction with dotted curves and solutions of the deterministic model (\ref{number_model_red}) with dashed curves. Because of the larger population, the infection is less prone to go
extinct and in many cases, we can continue the simulation at least until $t=1000$ (cf. Bartlett (1957)\nocite{Bartlett.JRSS_A:120}). Oscillations in the proportion of infected individuals remain after stabilization of the solution of (\ref{number_model_red}) near the endemic equilibrium. Note that the vertical axis has $\log_{10}$-scale so the amplitude of the oscillations of the infected proportion are larger than they seem at a first glance. The maxima correspond to proportions of infected of about $10^{-1.5}\approx .316$ whereas the minima correspond to proportions of about $10^{-2.2}\approx .0063$. The first maximum and its subsequent minimum correspond to an even higher amplitude, making it difficult for very infectious diseases to enter and remain endemic in small populations. Therefore, very infectious diseases either require large populations in order to remain endemic or evolve in a more infectious direction with time.

\begin{figure}
\epsfxsize=128mm
\begin{picture}(264,290)(0,0)
\put(0,0){\epsfbox{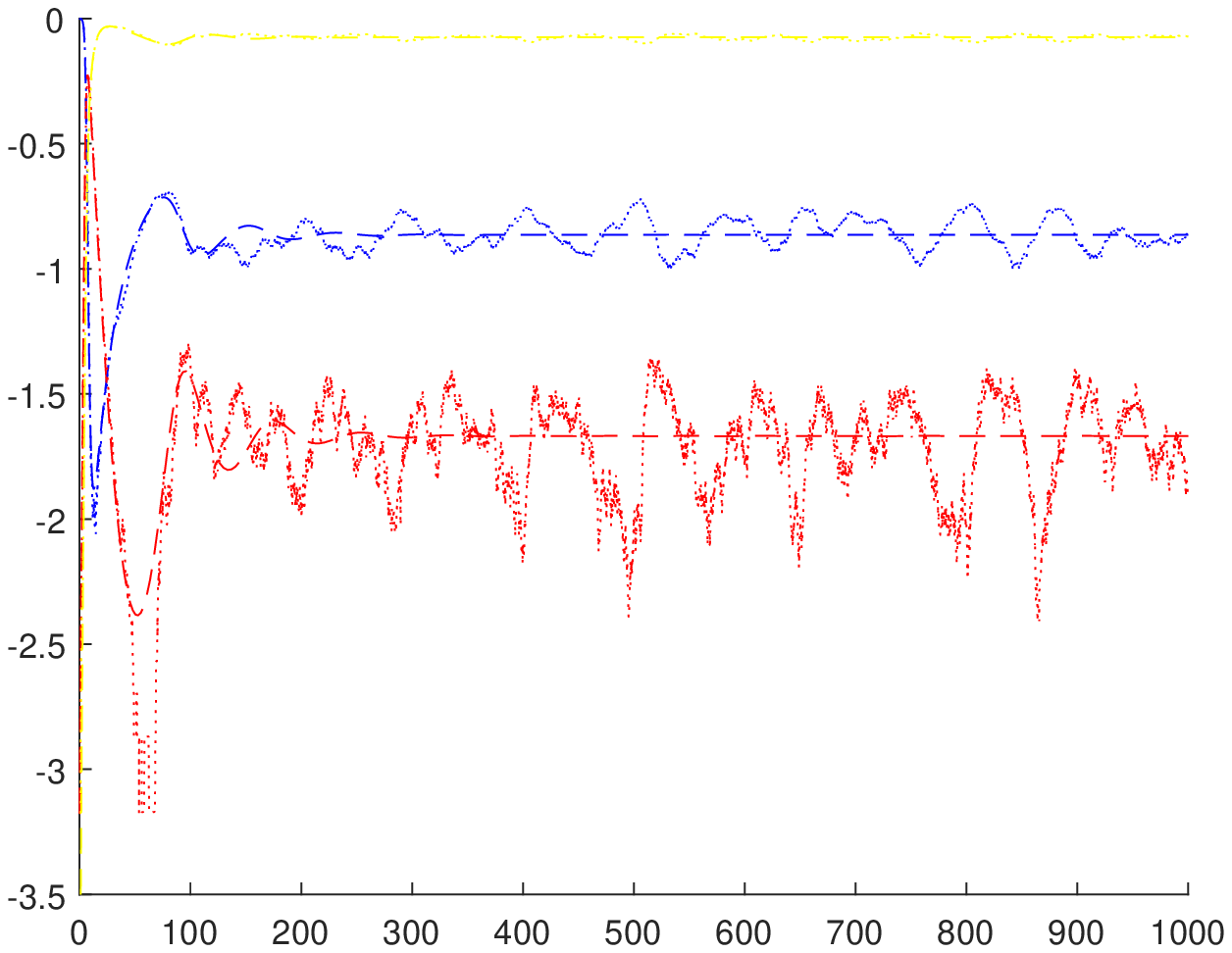}}
\put(-5,260){$\log_{10}$}
\put(335,0){$t$}
\end{picture}
\caption{$\log_{10}$ values of proportions of susceptible (blue), infected (red), and removed (yellow) individuals as predicted by (\ref{epid-diff-diff}) (dotted) and (\ref{number_model_red}) (dashed). Parameters used are
$\lambda=.0051$, $\mu=.005$, $\gamma=.2$, $\omega=0$, and $\beta=1.5$ and the used initial conditions are $S(0)=1499$, $I(0)=1$, and $R(0)=0$. It does not make sense to simulate frequency distributions of (\protect\ref{epid-diff-diff}) because any valid numerical model becomes far too large. The dashed curves are the corresponding simulation of the deterministic model (\ref{number_model_red}). We indicate a particular realization of an outcome from model (\protect\ref{epid-diff-diff}) by dotted curves.}
\label{probabilities_mc}
\end{figure}

\section{Vaccination and herd-immunity}
\label{vacc}

We have so far considered (A) social distancing measures basically changing the contact rate $\beta$, (B) quarantine and contact tracing measures removing possible infective individuals by changing the removal rate
$\gamma$, and (C) travel restriction measures isolating some part of the population. In terms of model (\ref{2D-model}) the measures (A) and (B) are both equivalent to a change in $\rho_0$ whereas the measure (C) hits the
long-run tendency that an epidemic will affect a quite small proportion of the population consisting of just a few individuals if the population remains isolated enough. Subsequent longer periods giving support for not more than
a few infected individuals in the population increases the probability of extinction of the disease rapidly.

We now sketch the situation for some cases with vaccines available, ie $\omega>0$ and limit this discussion to the deterministic model (\ref{2D-model}). The endemic equilibrium is globally asymptotically stable when a hyperbola (the susceptible isocline) intersects a vertical line (the infective isocline) in the interior of compact triangle $x\geq 0$, $y\geq 0$, $x+y\leq 1$. Social distancing measures (A) and quarantine and contact tracing measures (B) both shifted
the vertical line to the right but vaccination will shift the hyperbola to the left (removal of susceptibles) and will allow for the possibility of removing the disease by forming a globally stable equilibrium inside the herd immunity region
\begin{displaymath}
x\geq 0,\: y\geq 0,\: x+y\leq \frac{1}{\rho_0},
\end{displaymath}
see Figure \ref{picture_series2}(b). Such a possibility does not exist without vaccination. In fact, the herd immunity region is never defined if social distancing, quarantine, or contact tracing measures are used since they all change $\rho_0$. In the case of natural transmission, the concept of herd immunity is either irrelevant (Figure \ref{picture_series1}(a)), not well defined or not a domain of attraction
(Figure \ref{picture_series1}(b)), or sensitive to modeling assumptions and structural instability (Figure \ref{picture_series1}(c)).

\section{Summary}
\label{sum}

Much confusion about epidemic models have been prevailing over decades and we think it is time to describe the mechanism that is responsible for the subsequent transitions between stochastic and deterministic behavior in such models.

Our assumption of an exponential growth of the population and the fact that population parameters like natural birth and death operate on a slower time-scale than disease parameters like
contact and removal rates ensures that globally stable endemic equilibria correspond to low proportions of infected individuals. Low proportions of individuals correspond to low numbers of individuals if the population is isolated or
small. These low proportions decrease with increased pre-vaccination disease control measures aiming at control of the infectiousness of the disease.

Low asymptotic average proportions constitute a mechanism for a transition from primarily deterministic dynamics to stochastic dynamics. Subsequent waves start therefore,
usually as local outbreaks at seemingly randomly selected places.

Our description in terms of proportion makes another result clear, too. The strain with the best ability to increase its proportion outcompetes all other strains whenever the proportion of susceptible remain near an endemic equilibrium. The expected effect of this evolution theorem is
that the contact rate first increases until the pathogen has reached some physiological bound. After this, the pathogen attempts to decrease the removal rate by avoiding virulence and quarantine measures (cf. e. g. Smith (1887)\nocite{Smith.AMNat:21}).

The stochastic formulation of the model includes the possibility of extinction of both the population and the disease. If the size of the population is small, such probabilities are usually high. Diseases can go extinct on isolated
islands. We remind the reader that during the CoViD-19 epidemic the disease was easier to get under control at societies with easily controlled boarders (Iceland, Taiwan, Oceania, New Zeeland,
Australia, Japan, and South Korea). The smaller community, the less sensitive the result was to the set of disease control strategies actually applied.

The smaller the unit with easily controlled boarders is, the larger the probability of extinction. As soon as the population size becomes sufficiently large for keeping the probability of extinction low, oscillatory patterns that are
not present in the deterministic model become visible. It is therefore, essential to implement regional quarantine measures as soon as we observe local outbreaks.

The frequency distributions of the stochastic model do neither display any time-dependent nor oscillatory patterns. Diffusion approximations (see e. g Allen (2008)\nocite{Brauer_Driesche_cpt3})  remain close enough to ensure sufficient stability. Instead, it is the particular realizations of the model that display the oscillatory pattern. These oscillations might introduce violations of the competitive exclusion principle alluded to above but this is not very likely.
The formulation of useful conditions for such violations remains an open problem.

We end up by a discussion of the herd-immunity concept (Plans-Rubi{\'{o}} (2021)\nocite{Plans-Rubio.Vaccines:9}) and conclude that this concept remains without meaning in connection to all pre-vaccination disease control strategies (cf. e. g Jones and Helmreich (2020)\nocite{Lancet.Jones:396}). If we have $\rho_0\leq 1$, then we are automatically in the herd-immunity region. The concept is irrelevant. If we have $\rho_0>1$, then the herd-immunity region is not an attraction domain and the average proportion of infected individuals will tend
to remain unchanged with time. Natural transmission cannot produce herd-immunity for eradication of the disease. In addition, the objective of pre-vaccination disease control strategies is to make $\rho_0$ smaller. Such measures
change the herd-immunity region and introduce problems with the definition of such a concept.

Finally, we comment the concept of herd-immunity in the case a disease is so infectious that it will infect almost the whole population in the first wave (Kermack and McKendrick (1927)\nocite{Kermack.ProcCRSocLonA:115}).
In this case, we are close to a structurally unstable case (Guckenheimer and Holmes (1983)\nocite{guck}) and the number of infected individuals will be extremely low after the first wave, too. The probability of extinction of
the disease is after such an event huge. Measles is a very infectious disease that could be practically extinct during the pre-vaccination era for many years at islands containing populations of up to some hundred thousand individuals (Bartlett (1957)\nocite{Bartlett.JRSS_A:120}, Black (1966)\nocite{Black.JTB:11}, and Cliff and Haggett (1980)\nocite{Cliff.JHygiene:85}). However, measles was never eradicated. We may naturally discuss the spectrum of phenomena that we would like to include in a herd immunity concept. However, we think that we create a mess around the concept if we include population size in the concept. If herd-immunity is used then we must restrict it to relate to well-defined proportions of removed individuals only. Such criteria hold only when we are able to apply vaccination strategies successfully enough to make all other disease control measures superfluous.

Evolution makes diseases more infectious with time (Theorem \ref{glob-stab-thm}) and it sounds natural to ask how infections they be at the end of this paper. Eventually, the infectiousness of the disease is most likely controlled by the population size (e. g (Bartlett (1957)\nocite{Bartlett.JRSS_A:120}). The reason is that the probability of extinction of the disease is larger for smaller populations. Figure \ref{probabilities} demonstrated an example of a disease with $R_0\approx 7.317$ and $\rho_0=7.314$. The probability of survival of the disease drops rapidly after $150$ time units for an initial population size of $210$ individuals. An equally infectious disease remains in the population almost certainly at least $1000$ time units for an initial population size consisting of $1500$ individuals, compare Figure \ref{probabilities_mc}.

\appendix

\section{Proof of Theorem \protect\ref{absorbing_state_theorem}}
\label{Appendproof_absorbing}

\begin{proof}
We have
\begin{displaymath}
\frac{dq_{0,0,0}}{dt}=\mu(q_{1,0,0}+q_{0,1,0}+q_{0,0,1})+\gamma q_{0,1,0}+\omega q_{1,0,0}\geq 0.
\end{displaymath}
and
\begin{eqnarray*}
\sum_{S=0}^\infty\sum_{R=0}^\infty\frac{dq_{S,0,R}}{dt}&=&\mu\sum_{R=0}^\infty q_{1,0,R}+\mu\sum_{R=0}^\infty q_{0,1,R}+\mu\sum_{R=0}^\infty (R+1)q_{0,0,R+1}\\
&&-\mu\sum_{R=0}^\infty Rq_{0,0,R}-\lambda\sum_{R=0}^\infty Rq_{0,0,R}+\gamma\sum_{R=0}^\infty q_{0,1,R}+\omega\sum_{R=0}^\infty q_{1,0,R}\\
&&+\mu\sum_{S=1}^\infty\sum_{R=0}^\infty(S+1)q_{S+1,0,R}+\mu\sum_{S=1}^\infty\sum_{R=0}^\infty q_{S,1,R}\\
&&+\mu\sum_{S=1}^\infty\sum_{R=0}^\infty(R+1)q_{S,0,R+1}-\mu\sum_{S=1}^\infty\sum_{R=0}^\infty(S+R)q_{S,0,R}\\
&&+\lambda\sum_{S=1}^\infty\sum_{R=0}^\infty(S+R-1)q_{S-1,0,R}-\lambda\sum_{S=1}^\infty\sum_{R=0}^\infty(S+R)q_{S,0,R}\\
&&+\gamma\sum_{S=1}^\infty\sum_{R=0}^\infty q_{S,1,R}+\omega\sum_{S=1}^\infty\sum_{R=0}^\infty(S+1)q_{S+1,0,R}\\
&&-\omega\sum_{S=1}^\infty\sum_{R=0}^\infty Sq_{S,0,R}\\
&=&\mu\sum_{R=0}^\infty q_{1,0,R}+\mu\sum_{R=0}^\infty q_{0,1,R}+\mu\sum_{R=1}^\infty R q_{0,0,R}-\mu\sum_{R=0}^\infty Rq_{0,0,R}\\
&&-\lambda\sum_{R=0}^\infty Rq_{0,0,R}+\gamma\sum_{R=0}^\infty q_{0,1,R}+\omega\sum_{R=0}^\infty q_{1,0,R}\\
&&+\mu\sum_{S=2}^\infty\sum_{R=0}^\infty Sq_{S,0,R}+\mu\sum_{S=1}^\infty\sum_{R=0}^\infty q_{S,1,R}+\mu\sum_{S=1}^\infty\sum_{R=1}^\infty Rq_{S,0,R}\\
&&-\mu\sum_{S=1}^\infty\sum_{R=0}^\infty(S+R)q_{S,0,R}+\lambda\sum_{S=0}^\infty\sum_{R=0}^\infty(S+R)q_{S,0,R}\\
&&-\lambda\sum_{S=1}^\infty\sum_{R=0}^\infty(S+R)q_{S,0,R}+\gamma\sum_{S=1}^\infty\sum_{R=0}^\infty q_{S,1,R}\\
&&+\omega\sum_{S=2}^\infty\sum_{R=0}^\infty Sq_{S,0,R}-\omega\sum_{S=1}^\infty\sum_{R=0}^\infty Sq_{S,0,R}\\
&=&\mu\sum_{R=0}^\infty q_{1,0,R}+\mu\sum_{S=0}^\infty\sum_{R=0}^\infty q_{S,1,R}+\gamma\sum_{S=0}^\infty\sum_{R=0}^\infty q_{S,1,R}\geq 0
\end{eqnarray*}
meaning that, the probability of extinction of the population and the disease will increase with time, respectively.
\end{proof}


\bibliographystyle{abbrv}
\bibliography{artiklar}
\end{document}